\documentclass[
final
]{dmtcs-episciences}

\usepackage[utf8]{inputenc}
\usepackage{subfigure}

\usepackage{amsfonts,amsmath}
\usepackage{amsthm}
\usepackage{amssymb}
\usepackage{url}
\usepackage[norelsize,lined,ruled,linesnumbered]{algorithm2e}
\usepackage{mathtools}
\mathtoolsset{showonlyrefs,showmanualtags}

\usepackage{color}

\newtheorem{lemma}{Lemma}

\newtheorem{theorem}{Theorem}
\newtheorem{corollary}{Corollary}
\newtheorem{conjecture}{Conjecture}

\usepackage{nicefrac}

\usepackage[T1]{fontenc}

\usepackage[round]{natbib}

\newcommand{\doi}[1]{doi: \href{https://doi.org/#1}{\nolinkurl{#1}}}

\author[Hiroshi Fujiwara et al.]{Hiroshi
  Fujiwara\affiliationmark{1}\thanks{This
  work was supported by JSPS KAKENHI Grant Numbers~25K14990,
  20K11689, and 20K11676;
  ROIS NII Open Collaborative Research 261FP-24163;
  the Support Center for Advanced Telecommunications Technology
  Research (SCAT);
  and the Research Institute for Mathematical Sciences (RIMS),
  an International Joint Usage/Research Center located in Kyoto
  University.}
  \and Kota Miyagi\affiliationmark{2}
  \and Katsuhisa Ouchi\affiliationmark{1}}
\title[Pinwheel Scheduling with Real Periods]{Pinwheel Scheduling with Real Periods}
\affiliation{
  Shinshu University, Nagano, Japan\\
  The University of Electro-Communications, Chofu, Japan}
\begin{document}
\publicationdata{vol. 28:4, SOFSEM 2026}{2026}{4}{10.46298/dmtcs.17657}{2026-03-06; 2026-03-06; 2026-07-07}{2026-07-07}
\maketitle
\begin{abstract}
  For a sequence of tasks, each with a positive integer period, the
  pinwheel scheduling problem involves finding a valid schedule in the
  sense that the schedule performs one task per day and each task is
  performed at least once every consecutive days of its period.  It
  had been conjectured by Chan and Chin (1993) that there exists a
  valid schedule for any sequence of tasks with density, the sum of
  the reciprocals of each period, at most $\frac{5}{6}$.  Recently,
  Kawamura (2024, in press) settled this conjecture affirmatively.  In this paper we
  consider an extended version with real periods proposed by Kawamura,
  in which a valid schedule must perform each task $i$ having a real
  period~$a_{i}$ at least $l$ times in any $\lceil l a_{i}
  \rceil$ consecutive days for all positive integer $l$.  We show that any
  sequence of tasks such that the periods take three distinct real
  values and the density is at most $\frac{5}{6}$ admits a valid
  schedule.  We hereby conjecture that the conjecture of Chan and Chin is
  true also for real periods.
\end{abstract}
\section{Introduction}
\label{sec:introduction}
In the pinwheel scheduling problem introduced by \cite{48075}, we are given a
sequence of tasks, each with a positive integer period, and want to
find a schedule that performs one task per day such that each task is
performed at least once every consecutive days of its period.  Let
$\mathbb{N}$ and $\mathbb{Z}$ be the set of all positive integers and
the set of all integers, respectively, and let $[k] = \{i \;|\; i \in
\mathbb{N}, 1 \leq i \leq k\}$.  Formally, an integer-valued instance
of the pinwheel scheduling problem is a sequence $A = (a_{i})_{i \in
  [k]} \in \mathbb{N}^{k}$, where $a_{i}$ is the \emph{period} of task
$i$.  The goal is to find a \emph{schedule} $S: \mathbb{Z} \to [k]$,
which indicates that task $S(t)$ is performed on day $t$, such that the
following condition is satisfied for each task $i \in [k]$:
\begin{quote}
  for each $m \in \mathbb{Z}$, there exists a day $t \in [m,
    m + a_{i}) \cap \mathbb{Z}$ such that $S(t) = i$.
\end{quote}
Such a schedule is said to be \emph{valid} for the instance $A$.  In
particular, when a schedule $S$ is periodic, that is, there exists $p
\in \mathbb{N}$ such that $S(t) = S(t + p)$ holds for all $t \in
\mathbb{Z}$, we represent the schedule by a sequence of $|S(0)S(1)
\cdots S(p - 1)|$.  We say that an instance is \emph{schedulable} if
there exists a valid schedule for the instance.  For example, the instance
$(2, 4, 4)$ is schedulable, since schedule $|1213|$ is valid for it.

The \emph{density} of an instance $A = (a_{i})_{i \in [k]}$ is defined as
$D(A) = \sum_{i \in [k]} \frac{1}{a_{i}}$.  Intuitively, the density
of an instance means how many tasks a valid schedule for the instance
has to perform on average per day.  Obviously, $D(A) \leq 1$ is
necessary for the instance $A$ to admit a valid schedule.  (See that the
density of instance $(2, 4, 4)$ above is 1.)  However, this is not
sufficient.  Chan and Chin conjectured as follows:
\begin{conjecture}[\cite{CC93}]
  If an integer-valued instance $A$
  satisfies $D(A) \leq \frac{5}{6}$, then $A$ is schedulable.
  \label{conj:integer}
\end{conjecture}
This conjecture is best possible, since for any $a_{3}$, the instance
$(2, 3, a_{3})$, whose density is $\frac{5}{6} + \frac{1}{a_{3}} >
\frac{5}{6}$, is not schedulable.  After much work by many researchers
such as
\cite{linlin97,10.1007/s00453-002-0938-9,doi:10.1137/1.9781611977042.8},
\cite{KawamuraPNAS} finally settled this conjecture in the affirmative.
Please refer to \cite{10.1145/3618260.3649757} for the conference version
of the paper.
\begin{theorem}[{\cite{10.1145/3618260.3649757,KawamuraPNAS}}]
  If an integer-valued instance $A$
  satisfies $D(A) \leq \frac{5}{6}$, then $A$ is schedulable.
  \label{thm:integer}
\end{theorem}

\cite{10.1145/3618260.3649757,KawamuraPNAS} proposed an
extended version of the pinwheel scheduling problem that allows
\emph{real} values at least 1 for periods of tasks, replacing the
validity condition for each task $i \in [k]$ by:
\begin{quote}
  for each $l \in \mathbb{N}$ and $m \in \mathbb{Z}$,
  there exist at least $l$ values of $t \in [m, m + \lceil l
  a_{i} \rceil) \cap \mathbb{Z}$ such that $S(t) = i$.
\end{quote}
In other words, this condition is that for each $l \in \mathbb{N}$,
task $i$ is performed at least $l$ times in any  $\lceil l
a_{i} \rceil$ consecutive days.  Note that when $a_{i}$ is an integer, the
condition is equivalent to that of the original integer-valued
problem.  For example, a task with period~$\frac{7}{2}$ has to be
performed at least once every 4 days and at least twice every 7 days.
The schedule $|1112112|$ is valid for the instance $(2, \frac{7}{2})$.

It should be remarked that this is not a simple relaxation of the
problem.  More specifically, even if one identifies $n$ tasks each
with period~$a$ in a valid schedule, the identified task may not be
performed according to a period of $\frac{a}{n}$.  For example,
the schedule $|1111213|$ is valid for the instance $(2, 7, 7)$.  Identify
tasks~2 and 3.  Then, the identified task fails to meet the validity
condition for a period of $\frac{7}{2}$; the task is not performed
once every 4 days, though it is performed twice every 7 days.

Nevertheless, we conjecture that Theorem~\ref{thm:integer} holds true
also for the real-valued version:
\begin{conjecture}
  If a real-valued instance $A$ satisfies $D(A) \leq \frac{5}{6}$,
  then $A$ is schedulable.
  \label{conj:real}
\end{conjecture}
This is also best possible for the same reason as the integer-valued
problem.  Our main theorem, Theorem~\ref{thm:real_3_distinct}, is a
partial solution to Conjecture~\ref{conj:real}, which strengthens the
plausibility of the conjecture.
\begin{theorem}
  If a real-valued instance $A$ consists of three distinct values and
  satisfies $D(A) \leq \frac{5}{6}$, then $A$ is schedulable.
  \label{thm:real_3_distinct}
\end{theorem}
\subsection{Previous Results}
\label{subsec:previous_results}
\paragraph{The integer-valued pinwheel scheduling problem.}
The next theorem is one of the results aimed at solving
Conjecture~\ref{conj:integer}, which can be regarded as an
integer-valued version of our main theorem.
\begin{theorem}[{\cite{linlin97}}]
  If an integer-valued instance $A$ consists of three distinct
  values and satisfies $D(A) \leq \frac{5}{6}$, then $A$ is
  schedulable.
  \label{thm:integer_3_distinct}
\end{theorem}
We will see below that it seems hopeless to extend
Theorem~\ref{thm:integer_3_distinct} directly to the real-valued
problem.  (Conversely, our proof of Theorem~\ref{thm:real_3_distinct}
can be seen as a simpler proof of
Theorem~\ref{thm:integer_3_distinct}.)  The proof of
Theorem~\ref{thm:integer_3_distinct} in \cite{linlin97}
constructs a valid schedule for a given integer-valued instance
\begin{equation*}
  (\underbrace{c_{1}, \cdots, c_{1}}_{q_{1}}, \underbrace{c_{2},
  \cdots, c_{2}}_{q_{2}}, \underbrace{c_{3}, \cdots, c_{3}}_{q_{3}}).
\end{equation*}
In the schedule, for each $i = 1, 2,$ and $3$, identify all the tasks
with period~$c_{i}$ and rename them as task~$i$.  If the resulting
schedule were valid for the instance $(\frac{c_{1}}{q_{1}},
\frac{c_{2}}{q_{2}}, \frac{c_{3}}{q_{3}})$, this would give a solution
to the problem with rational periods.  However, this attempt fails.
We have a counterexample: the instance 
\begin{equation*}
  (\underbrace{13, \cdots, 13}_{5},
  \underbrace{14, \cdots, 14}_{5}, \underbrace{78, \cdots, 78}_{5}).
\end{equation*}
Run the construction for this instance and identify all the
5 tasks with period~14 in the resulting schedule.
It will turn out
that the identified task is not performed at a period of
$\frac{14}{5}$.

\cite{doi:10.1137/1.9781611977042.8} partially solved
Conjecture~\ref{conj:integer} for the case of 12 or fewer tasks.
The authors provided a
\emph{Pareto surface} of the whole set of such instances, where a
Pareto surface $C$ for a set of pinwheel scheduling instances $I$ is
an inclusion minimal set of pairs of an instance with minimum periods
and a schedule that is valid for the instance such that, for every
instance $A$ in $I$, $C$ includes a valid schedule for $A$.  We will
also consider this point of view in Sections~\ref{sec:2_periods_proof}
and~\ref{sec:3_periods_proof}.

We introduce a stronger conjecture recently made by Kawamura:
\begin{conjecture}[{\cite{Kawamura2025}}]
  If an integer-valued instance $A = (a_{i})_{i \in [k]}$ satisfies
  $D(A) \leq 1 - \frac{a_{1} - 1}{a_{1}(a_{1} + 1)}$, then $A$ is
  schedulable.
  \label{conj:real_a_1}
\end{conjecture}
\paragraph{The real-valued pinwheel scheduling problem.}
\cite{10.1145/3618260.3649757,KawamuraPNAS} extended the pinwheel
scheduling problem to real periods as one of the techniques for
proving Theorem~\ref{thm:integer}.  The following theorem is a partial
solution to Conjecture~\ref{conj:real}.
\begin{theorem}[{\cite{10.1145/3618260.3649757,KawamuraPNAS}}]
  If a real-valued instance $A$ contains at most two periods
  and satisfies $D(A) \leq 1$, then $A$ is schedulable.
  \label{thm:real_2_density_1}
\end{theorem}
Theorem~\ref{thm:real_a_1} below can be seen as a partial solution to
Conjecture~\ref{conj:real} for the case where $a_{1} > 103$, since the
density threshold in the statement exceeds $\frac{5}{6}$ in this case.
\begin{theorem}[{\cite{10.1007/978-3-031-92935-9_12}}]
  If a real-valued instance $A = (a_{i})_{i \in [k]}$ satisfies $D(A)
  \leq 1 - \frac{1 + \ln 2}{1 + \sqrt{a_{1}}} - \frac{3}{2 a_{1}}$,
  then $A$ is schedulable.
  \label{thm:real_a_1}
\end{theorem}
\paragraph{Other related problems.}
The pinwheel scheduling problem can be viewed as a task packing
problem.  Its dual problem, called the pinwheel \emph{covering}
problem, has recently been studied by
\cite{KAWAMURA2020195,10.1007/978-3-031-92935-9_12}.  The bamboo
garden trimming problem~(\cite{GASIENIEC2024103476,KawamuraPNAS}) is
an optimization problem closely related to the pinwheel scheduling
problem.  Refer to \cite{Kawamura2025} for further variants and
related problems.
\section{A Warm-Up: Case of Two Distinct Real Values}
\label{sec:2_periods_proof}
To demonstrate the idea of the proof of our main theorem,
Theorem~\ref{thm:real_3_distinct}, for a simpler case, we first
attempt to prove a weaker theorem for the case of two distinct real
values, Theorem~\ref{thm:real_2_distinct}.  We mention that this can
be obtained as a corollary of Theorem~\ref{thm:real_2_density_1} and
Lemma~\ref{lem:partitioning} below.
\begin{theorem}
  If a real-valued instance $A$ consists of two distinct values
  and satisfies $D(A) \leq \frac{5}{6}$, then $A$ is schedulable.
  \label{thm:real_2_distinct}
\end{theorem}
The following two lemmas play an important role.  All schedulability
throughout this paper, except for that of some individual instances,
is derived from these lemmas.
Note in Lemma~\ref{lem:monotonicity} that
the same schedule remains valid for the new instance.
\begin{lemma}[{\cite{10.1145/3618260.3649757,KawamuraPNAS}}]
  If a schedule $S$ is valid for a real-valued instance $(a_{i})_{i
    \in [k]}$, then $S$ is valid also for the instance obtained by
  replacing some period~$a_{j}$ ($j \in [k]$) by a real $b \geq
  a_{j}$.
  \label{lem:monotonicity}
\end{lemma}
\begin{lemma}[{\cite{10.1145/3618260.3649757,KawamuraPNAS}}]
  If a real-valued instance $(a_{i})_{i \in [k]}$ is schedulable,
  so is the instance obtained by
  replacing some period~$a_{j}$ ($j \in [k]$) by $q$ copies of $q
  a_{j}$ for any $q \in \mathbb{N}$.
  \label{lem:partitioning}
\end{lemma}
Lemma~\ref{lem:real_2_exact} below states that any real-valued
instance with two periods and a density of exactly $\frac{5}{6}$ is
schedulable.  Instead of such an instance $(a_{1},
\frac{1}{\frac{5}{6} - \frac{1}{a_{1}}})$ itself, we focus on the
value $a_{1} \in \mathbb{R}$ as the projection of the instance onto
the real line.  Note that the condition $\frac{6}{5} < a_{1} \leq
\frac{12}{5}$ in the statement ensures, without loss of generality,
that the periods are ordered in non-decreasing order.
Lemmas~\ref{lem:monotonicity}, \ref{lem:partitioning}, and
\ref{lem:real_2_exact} immediately imply
Theorem~\ref{thm:real_2_distinct}.
\begin{lemma}
  Every real-valued instance $A = (a_{1}, \frac{1}{\frac{5}{6} -
    \frac{1}{a_{1}}})$ with $\frac{6}{5} < a_{1} < \frac{12}{5}$
  satisfies at least one of the following conditions and the schedule
  $S$ defined there is valid for the instance $A$:

  (i)~If $\frac{6}{5} < a_{1} \leq \frac{3}{2}$, then,
  \begin{alignat}{2}
    S(t) & =
    \begin{cases}
      1,           &  t \equiv 0, 1, 2, 3, 4;\\
      2,           &  t \equiv 5 \;\; (\text{mod } 6),
    \end{cases}
    \label{eq:miyagi_2_1}
  \end{alignat}
  which is represented as $|111112|$.

  (ii)~If $\frac{3}{2} \leq a_{1} \leq 2$, then,
  \begin{alignat}{2}
    S(t) & =
    \begin{cases}
      1,           &  t \equiv 0, 1;\\
      2,           &  t \equiv 2 \;\; (\text{mod } 3),
    \end{cases}
    \label{eq:miyagi_2_2}
  \end{alignat}
  which is represented as $|112|$.

  (iii)~If $2 \leq a_{1} \leq 3$, then,
  \begin{alignat}{2}
    S(t) & =
    \begin{cases}
      1,           &  t \equiv 0;\\
      2,           &  t \equiv 1 \;\; (\text{mod } 2),
    \end{cases}
    \label{eq:miyagi_2_3}
  \end{alignat}
  which is represented as $|12|$.
  \label{lem:real_2_exact}
\end{lemma}
It does not matter that the ranges of $a_{1}$ in the cases~(i), (ii),
and (iii) are not exclusive and even include values out of
$\frac{6}{5} < a_{1} \leq \frac{12}{5}$.  The former fact
means that there may exist multiple valid schedules for a single
instance.

Now we are going to provide some lemmas for the proof of
Lemma~\ref{lem:real_2_exact}.  We have found by hand that the three
instances each with a density of~1 in the following
Lemma~\ref{lem:miyagi_2} are schedulable and each period is minimum
possible.
\begin{lemma}
  The schedules defined by~\eqref{eq:miyagi_2_1},
  \eqref{eq:miyagi_2_2}, and \eqref{eq:miyagi_2_3} are valid for
  the
  real-valued instances $(\frac{6}{5}, 6)$, $(\frac{3}{2}, 3)$, and
  $(2, 2)$, respectively.
  \label{lem:miyagi_2}
\end{lemma}
\begin{proof}
  It is immediate to see that each of the schedules satisfies the
  condition for the tasks having an integer period.  All that remains
  is to confirm the schedulability of the tasks having a rational
  period.

  We first demonstrate that $S$ in~\eqref{eq:miyagi_2_1} satisfies the
  condition for task~1 with period~$\frac{6}{5}$ in the instance
  $(\frac{6}{5}, 6)$.  We denote by $P_{i}(y, l)$ the proposition that
  for each $m \in \mathbb{Z}$, there exist at least $l$ values of $t
  \in [m, m + y) \cap \mathbb{Z}$ such that $S(t) = i$, which means
    that the schedule $S$ performs task $i$ at least $l$ times in any
    $y$ consecutive days.  The schedulability of task~1 in the
    instance $(\frac{6}{5}, 6)$ is represented as $P_{1}(\lceil l
    \cdot \frac{6}{5} \rceil, l)$ for all $l \in \mathbb{N}$. We are
    going to prove this.

  For each $l = 1, 2, \ldots, 5$, we have $\lceil l \cdot \frac{6}{5}
  \rceil = l + 1$.  According to \eqref{eq:miyagi_2_1}, $S(t) = 1$ if
  $t \equiv 0, 1, 2, 3, 4$ $(\text{mod } 6)$.
  It is clear that every $l + 1$ consecutive days,
  task~1 is performed at least $l$ times.
  That is, $P_{1}(\lceil l \cdot \frac{6}{5} \rceil, l)$ is true.

  For $l \geq 6$, let $q$ be the quotient and $r$ be the remainder
  when $l$ is divided by 5. We have $\lceil l \cdot \frac{6}{5} \rceil
  = \lceil (5q + r) \cdot \frac{6}{5} \rceil = 6q + \lceil r \cdot
  \frac{6}{5} \rceil$.  Using what we have already shown, we obtain that:
  Task~1 is
  done at least $r$ times in any $\lceil r \cdot
  \frac{6}{5} \rceil$ consecutive days by $P_{1}(\lceil z \cdot \frac{6}{5}
  \rceil, z)$ for $z = 1, 2, 3, 4$, while it is done at least $5q$
  times in any $6q$ consecutive days by $P_{1}(\lceil 5 \cdot
  \frac{6}{5} \rceil, 5)$.
  In total, task~1 is done at least $5q + r
  = l$ times in any $6q + \lceil r \cdot \frac{6}{5} \rceil$ consecutive days,
  which implies $P_{1}(\lceil l \cdot \frac{6}{5} \rceil, l)$.

  We next check
  task~1 with period~$\frac{3}{2}$ in the instance $(\frac{3}{2}, 3)$
  for $S$ in~\eqref{eq:miyagi_2_2} in the same way.
  Using the notation $P_{i}(y, l)$ again,
  we can state that our goal is to show
  $P_{1}(\lceil l \cdot \frac{3}{2} \rceil, l)$ for all $l \in
  \mathbb{N}$.

  For $l = 1$ or $2$, we have $\lceil l \cdot \frac{3}{2} \rceil = l +
  1$.  According to \eqref{eq:miyagi_2_2}, $S(t) = 1$ if $t \equiv 0,
  1$ $(\text{mod } 3)$.
  Clearly, every $l + 1$ consecutive
  days, task~1 is performed at least $l$ times.  That is,
  $P_{1}(\lceil l \cdot \frac{3}{2} \rceil, l)$ is true.

  For $l \geq 3$, let $q$ be the quotient and $r$ be the remainder
  when $l$ is divided by 2. We have $\lceil l \cdot \frac{3}{2} \rceil
  = \lceil (2q + r) \cdot \frac{3}{2} \rceil = 3q + \lceil r \cdot
  \frac{3}{2} \rceil$.  We obtain that: Task~1 is done at least $r$
  times in any $\lceil r \cdot \frac{3}{2} \rceil$ consecutive days by
  $P_{1}(\lceil 1 \cdot \frac{3}{2} \rceil, 1)$, while it is done at
  least $2q$ times in any $3q$ consecutive days by $P_{1}(\lceil 2
  \cdot \frac{3}{2} \rceil, 2)$.  Therefore, task~1 is done at least
  $2q + r = l$ times in any $3q + \lceil r \cdot \frac{3}{2} \rceil$
  consecutive days, which implies $P_{1}(\lceil l \cdot \frac{3}{2}
  \rceil, l)$.
\end{proof}

Using the following helper lemma, we can further derive a set of
schedulable instances from those found to be schedulable by
Lemma~\ref{lem:miyagi_2}.  Note that the resulting set includes
instances whose periods are not in non-decreasing order.
The lemma follows immediately from Lemma~\ref{lem:monotonicity}.
\begin{lemma}
  If a schedule $S$ is valid for
  the
  real-valued instance $(b_{1}, b_{2})$,
  then $S$ is valid also for
  the
  real-valued instance $(a_{1}, \frac{1}{\frac{5}{6} -
    \frac{1}{a_{1}}})$ with $a_{1} \geq b_{1}$ and
  $\frac{1}{\frac{5}{6} - \frac{1}{a_{1}}} \geq b_{2}$.
  \label{lem:real_2_cover}
\end{lemma}
\begin{proof}[of Lemma~\ref{lem:real_2_exact}]
  Applying Lemma~\ref{lem:real_2_cover}, the schedulability of the
  instance $(\frac{6}{5}, 6)$ by Lemma~\ref{lem:miyagi_2} leads us to
  the fact that the schedule defined by~\eqref{eq:miyagi_2_1} is valid
  for the instance $(a_{1}, \frac{1}{\frac{5}{6} - \frac{1}{a_{1}}})$
  if $a_{1} \geq \frac{6}{5}$ and $\frac{1}{\frac{5}{6} -
    \frac{1}{a_{1}}} \geq 6$, equivalently,
  \begin{equation}
    \frac{6}{5} < a_{1} \leq \frac{3}{2}.
    \label{eq:real_2_exact_1}
  \end{equation}
  Similarly, we derive that the schedule defined
  by~\eqref{eq:miyagi_2_2} is valid for the instance $(a_{1},
  \frac{1}{\frac{5}{6} - \frac{1}{a_{1}}})$ if
  \begin{equation}
    \frac{3}{2} \leq a_{1} \leq 2,
    \label{eq:real_2_exact_2}
  \end{equation}
  and that the schedule defined by~\eqref{eq:miyagi_2_3} is valid for
  the instance $(a_{1}, \frac{1}{\frac{5}{6} - \frac{1}{a_{1}}})$ if
  \begin{equation}
    2 \leq a_{1} \leq 3.
    \label{eq:real_2_exact_3}
  \end{equation}
  It is obvious that $a_{1}$ such that $\frac{6}{5} <
  a_{1} < \frac{12}{5}$ satisfies at least one of
  \eqref{eq:real_2_exact_1}, \eqref{eq:real_2_exact_2}, or
  \eqref{eq:real_2_exact_3}.
\end{proof}

We focus on real-valued instances with two periods.  The proof of
Theorem~\ref{thm:real_2_distinct} based on
Lemma~\ref{lem:real_2_exact} tells us that for any real-valued
instance with two periods and density less than or equal to
$\frac{5}{6}$, the schedule defined by either \eqref{eq:miyagi_2_1},
\eqref{eq:miyagi_2_2}, or \eqref{eq:miyagi_2_3} is valid.  Note that
each schedule is instance-insensitive in the sense that the same
schedule remains valid as long as either \eqref{eq:real_2_exact_1},
\eqref{eq:real_2_exact_2}, or \eqref{eq:real_2_exact_3} holds for the
instance.

Consider the above results in terms of the Pareto surface proposed by
\cite{doi:10.1137/1.9781611977042.8}.  They showed that $\{((2, 2),
|12|)\}$ is a Pareto surface of the set of all \emph{integer-valued}
instances with two periods and density less than or equal to
$\frac{5}{6}$.  In contrast, we have the following corollary for
\emph{real-valued} instances.  One can see that adding two
instance-schedule pairs to the Pareto surface for integer-valued
instances yields a Pareto surface for real-valued instances.
\begin{corollary}
  $\{((\frac{6}{5}, 6), |111112|), ((\frac{3}{2}, 3), |112|), ((2, 2),
  |12|)\}$ is a Pareto surface of the set of all real-valued instances
  with two periods and density less than or equal to $\frac{5}{6}$.
  \label{cor:pareto_2}
\end{corollary}
\section{Proof of Theorem~\ref{thm:real_3_distinct}}
\label{sec:3_periods_proof}
Theorem~\ref{thm:real_3_distinct} is proved in essentially the same
way as Theorem~\ref{thm:real_2_distinct}.  The following
Lemma~\ref{lem:real_3_exact} is the heart of the argument.  It states
that any real-valued instance with three periods and a density of exactly
$\frac{5}{6}$ is schedulable.  We deal with a pair $(a_{1}, a_{2})$ as
a projection of such an instance.  Note that the conditions $a_{1} <
a_{2}$ and $\frac{5}{6} \leq \frac{1}{a_{1}} + \frac{2}{a_{2}}$ in the
statement guarantee, without loss of generality, that the periods are
ordered in non-decreasing order.  Lemmas~\ref{lem:monotonicity},
\ref{lem:partitioning}, and \ref{lem:real_3_exact} immediately imply
Theorem~\ref{thm:real_3_distinct}.
\begin{lemma}
  Every real-valued instance $A = (a_{1}, a_{2}, \frac{1}{\frac{5}{6}
    - \frac{1}{a_{1}} - \frac{1}{a_{2}}})$ with $a_{1} < a_{2}$,
  $\frac{5}{6} < \frac{1}{a_{1}} + \frac{2}{a_{2}}$, and
  $\frac{1}{a_{1}} + \frac{1}{a_{2}} < \frac{5}{6}$ satisfies at least
  one of the following conditions and the schedule $S$ defined there
  is valid for the instance $A$:

  (I)~If $\frac{1}{a_{1}} +  \frac{2}{a_{2}} \leq 1$, then,
  letting $a_{2}' = \frac{2}{1 - \frac{1}{a_{1}}}$,
  \begin{alignat}{2}
    S(t) & =
    \begin{cases}
      1,           &  t \in
      \{\lceil j a_{1} \rceil - 1\;|\; j \in \mathbb{Z}\};\\
      2,           &  t \in
      \{\lfloor j a_{2}' \rfloor \;|\; j \in \mathbb{Z}\};\\
      3,           &  t \in
      \{\lfloor (j + \frac{1}{2}) a_{2}' \rfloor \;|\;
      j \in \mathbb{Z}\}.
      \label{eq:miyagi_3_1}
    \end{cases}
  \end{alignat}

  (II)~If $a_{1} \geq \frac{3}{2}$, $a_{2} > 5$, and $\frac{5}{6} -
  \frac{1}{9} \leq \frac{1}{a_{1}} + \frac{1}{a_{2}} < \frac{5}{6}$
  then
  \begin{alignat}{2}
    S(t) & =
    \begin{cases}
      1,           &  t \equiv 0, 1, 3, 4, 6, 7;\\
      2,           &  t \equiv 2, 5;\\
      3,           &  t \equiv 8 \;\; (\text{mod } 9),
    \end{cases}
    \label{eq:miyagi_3_2}
  \end{alignat}
  which is represented as $|112112113|$.

  (III)~If $a_{1} \geq \frac{11}{7}$, $a_{2} > 4$, and $\frac{5}{6} -
  \frac{1}{11} \leq \frac{1}{a_{1}} + \frac{1}{a_{2}} < \frac{5}{6}$
  then
  \begin{alignat}{2}
    S(t) & =
    \begin{cases}
      1,           &  t \equiv 0, 1, 3, 4, 6, 7, 9;\\
      2,           &  t \equiv 2, 5, 8;\\
      3,           &  t \equiv 10 \;\; (\text{mod } 11),
    \end{cases}
    \label{eq:miyagi_3_3}
  \end{alignat}
  which is represented as $|11211211213|$.

  (IV)~If $a_{1} \geq \frac{12}{7}$, $a_{2} > 3$, and $\frac{5}{6} -
  \frac{1}{12} \leq \frac{1}{a_{1}} + \frac{1}{a_{2}} < \frac{5}{6}$
  then,
  \begin{alignat}{2}
    S(t) & =
    \begin{cases}
      1,           &  t \equiv 0, 2, 3, 5, 7, 8, 10;\\
      2,           &  t \equiv 1, 4, 6, 9;\\
      3,           &  t \equiv 11 \;\; (\text{mod } 12),
    \end{cases}
    \label{eq:miyagi_3_4}
  \end{alignat}
  which is represented as $|121121211213|$.

  (V)~If $a_{1} > 2$, $a_{2} \geq 3$, and $\frac{5}{6} - \frac{1}{6}
  \leq \frac{1}{a_{1}} + \frac{1}{a_{2}} < \frac{5}{6}$ then
  \begin{alignat}{2}
    S(t) & =
    \begin{cases}
      1,           &  t \equiv 0, 2, 3;\\
      2,           &  t \equiv 1, 4;\\
      3,           &  t \equiv 5 \;\; (\text{mod } 6),
    \end{cases}
    \label{eq:miyagi_3_5}
  \end{alignat}
  which is represented as $|121123|$.

  (VI)~If $a_{1} > 2$, $a_{2} \geq \frac{12}{5}$, and $\frac{5}{6} -
  \frac{1}{12} \leq \frac{1}{a_{1}} + \frac{1}{a_{2}} < \frac{5}{6}$,
  then
  \begin{alignat}{2}
    S(t) & =
    \begin{cases}
      1,           &  t \equiv 0, 2, 4, 5, 7, 9;\\
      2,           &  t \equiv 1, 3, 6, 8, 10;\\
      3,           &  t \equiv 11 \;\; (\text{mod } 12),
    \end{cases}
    \label{eq:miyagi_3_6}
  \end{alignat}
  which is represented as $|121211212123|$.

  (VII)~If $a_{1} \geq \frac{12}{5}$, $a_{2} \geq \frac{12}{5}$, and
  $\frac{5}{6} - \frac{1}{6} \leq \frac{1}{a_{1}} + \frac{1}{a_{2}} <
  \frac{5}{6}$ then
  \begin{alignat}{2}
    S(t) & =
    \begin{cases}
      1,           &  t \equiv 0, 2, 4, 7, 9;\\
      2,           &  t \equiv 1, 3, 6, 8, 10;\\
      3,           &  t \equiv 5, 11 \;\; (\text{mod } 12),
    \end{cases}
    \label{eq:miyagi_3_7}
  \end{alignat}
  which is represented as $|121213212123|$.
  \label{lem:real_3_exact}
\end{lemma}
Lemma~\ref{lem:real_3_exact} splits the whole set of real-valued
instances with three periods and a density of exactly $\frac{5}{6}$ into
seven cases~(I) through~(VII).  In particular, for the analysis of
cases~(II) through~(VII), we discovered six real-valued instances of
Lemma~\ref{lem:miyagi_3_2-7} for which the covering of
Lemma~\ref{lem:real_3_inclusion} succeeds by trial and error while
checking schedulability by an integer programming solver.

We discuss some lemmas for the proof of
Lemma~\ref{lem:real_3_exact}.
We begin with instances that satisfy
the condition of case~(I) of Lemma~\ref{lem:real_3_exact}, that is
$\frac{1}{a_{1}} + \frac{2}{a_{2}} \leq 1$.  We here insist on a valid
schedule in an explicit form of $S(t) = \ldots$.  Otherwise, the
schedulability itself can easily be shown by
Lemmas~\ref{lem:monotonicity} and \ref{lem:partitioning}, and the
schedulability of the instance $(a_{1}, \frac{a_{2}}{2})$ implied by
Theorem~\ref{thm:real_2_density_1}.
\begin{lemma}
  The schedule $S$ defined by~\eqref{eq:miyagi_3_1} is valid for
  the
  real-valued instance $(a_{1}, a_{2}, \frac{1}{\frac{5}{6} -
    \frac{1}{a_{1}} - \frac{1}{a_{2}}})$ for any $a_{1}$ and $a_{2}$
  such that $a_{1} < a_{2}$, $\frac{5}{6} < \frac{1}{a_{1}} +
  \frac{2}{a_{2}} \leq 1$, and $\frac{1}{a_{1}} + \frac{1}{a_{2}} <
  \frac{5}{6}$.
  \label{lem:miyagi_1}
\end{lemma}
\begin{proof}
  First, we confirm that $S$ performs a single task for every day.
  Since $a_{2}' = \frac{2}{1 - \frac{1}{a_{1}}} > 2$, the set of days
  when task~2 is performed $\{\lfloor k a_{2}' \rfloor \;|\; k \in
  \mathbb{Z}\}$ and the set of days when task~3 is performed
  $\{\lfloor (k + \frac{1}{2}) a_{2}' \rfloor \;|\; k \in
  \mathbb{Z}\}$ have no intersection. Their union is $\{\lfloor k
  \cdot \frac{a_{2}'}{2} \rfloor \;|\; k \in \mathbb{Z}\}$.
  (See that $S$ for tasks~2 and 3 is constructed by treating them as a
  single task with period~$\frac{a_{2}'}{2}$ and then
  assigning tasks~2 and 3 in turn.)
  Moreover, for $a_{1}$ and $a_{2}'$ such that $\frac{1}{a_{1}} +
  \frac{2}{a_{2}'} = 1$, $\{\lfloor k \cdot \frac{a_{2}'}{2} \rfloor
  \;|\; k \in \mathbb{Z}\}$ and the set of days when task~1 is
  performed $\{\lceil j a_{1} \rceil - 1\;|\; j \in \mathbb{Z}\}$ are
  known to be a partition of
  $\mathbb{Z}$~(\cite{5638c72d-0b69-361b-b406-ab76ccc15749}).  (When
  $a_{1}$ and $a_{2}'$ are irrational numbers, this partition is
  called Rayleigh's or Beatty's theorem.)

  Next, we show that $S$ is valid for the instance $(a_{1}, a_{2}',
  a_{2}')$, which is sufficient to prove the lemma due to $a_{2} \geq
  a_{2}'$, $\frac{1}{\frac{5}{6} - \frac{1}{a_{1}} - \frac{1}{a_{2}}}
  \geq a_{2}'$, and Lemma~\ref{lem:monotonicity}.

  According to~\eqref{eq:miyagi_3_1}, for any $l \in \mathbb{N}$,
  the maximum time span in which task~1 is performed exactly $(l - 1)$
  times is from day $\lceil m a_{1} \rceil$ to day $(\lceil (m + l)
  a_{1} \rceil - 2)$ for some integer $m$, and its length is $(\lceil
  (m + l) a_{1} \rceil - \lceil m a_{1} \rceil - 1)$ days.  (i) If $l
  a_{1}$ is an integer, we evaluate $\lceil (m + l) a_{1} \rceil -
  \lceil m a_{1} \rceil - 1 = m a_{1} + \lceil l a_{1} \rceil - m
  a_{1} - 1 = l a_{1} - 1 = \lceil l a_{1} \rceil - 1$.  (ii)
  Otherwise, we get $\lceil (m + l) a_{1} \rceil - \lceil m a_{1}
  \rceil - 1 \leq (\lceil m a_{1} \rceil + \lceil l a_{1} \rceil) -
  \lceil m a_{1} \rceil - 1 = \lceil l a_{1} \rceil - 1$.  In any
  case, we observe that task~1 is done at least $l$ times in any
  $\lceil l a_{1} \rceil$ consecutive days.

  We finally check the validity of tasks~2 and 3 together.  Set $d =
  0$ for task~2 and $\frac{1}{2}$ for task~3.
  For any $l \in \mathbb{N}$,
  the maximum time span in which the task is performed
  exactly $(l - 1)$ times is from day $(\lfloor (m + d) a_{2}' \rfloor
  + 1)$ to day $(\lfloor (m + d + l) a_{2}' \rfloor - 1)$ for some
  integer $m$, and its length is $(\lfloor (m + d + l) a_{2}' \rfloor
  - \lfloor (m + d) a_{2}' \rfloor - 1)$ days.  (i) If $l a_{2}'$ is
  an integer, we obtain $\lfloor (m + d + l) a_{2}' \rfloor - \lfloor
  (m + d) a_{2}' \rfloor - 1 = \lfloor (m + d) a_{2}' \rfloor + l
  a_{2}' - \lfloor (m + d) a_{2}' \rfloor - 1 = l a_{2}' - 1 = \lceil
  l a_{2}' \rceil - 1$.  (ii) Otherwise, we get $\lfloor (m + d + l)
  a_{2}' \rfloor - \lfloor (m + d) a_{2}' \rfloor - 1 \leq (\lfloor (m
  + d) a_{2}' \rfloor + \lfloor l a_{2}' \rfloor + 1) - \lfloor (m +
  d) a_{2}' \rfloor - 1 = \lfloor l a_{2}' \rfloor = \lceil l a_{2}'
  \rceil - 1$.  In any case, we see that the task is done at least
  $l$ times in any $\lceil l a_{2}' \rceil$ consecutive days.
\end{proof}

The following lemma ensures the schedulability of the six real-valued
instances.  Each of the periods, except those involving $\varepsilon$,
is minimum possible.  Substituting zero for $\varepsilon$ makes the
instance non-schedulable.
\begin{lemma}
  The schedules defined by~\eqref{eq:miyagi_3_2}, \eqref{eq:miyagi_3_3},
  \eqref{eq:miyagi_3_4}, \eqref{eq:miyagi_3_5}, and \eqref{eq:miyagi_3_6}
  are valid for
  the
  real-valued instances
  $(\frac{3}{2}, 5 + \varepsilon, 9)$,
  $(\frac{11}{7}, 4 + \varepsilon, 11)$,
  $(\frac{12}{7}, 3 + \varepsilon, 12)$,
  $(2 + \varepsilon, 3, 6)$, and
  $(2 + \varepsilon, \frac{12}{5}, 12)$
  for any $\varepsilon > 0$, respectively.
  The schedule defined by~\eqref{eq:miyagi_3_7} is valid for
  $(\frac{12}{5}, \frac{12}{5}, 6)$.
  \label{lem:miyagi_3_2-7}
\end{lemma}
\begin{proof}
  It is easy to check for the tasks with integer periods.  For the
  tasks with rational periods, we can verify the validity as the same
  as in the proof of Lemma~\ref{lem:miyagi_2}.  We therefore focus on
  the tasks with periods that involve $\varepsilon$.  We will
  demonstrate that for any $\varepsilon > 0$, $S$
  in~\eqref{eq:miyagi_3_2} satisfies the condition for task~2 with
  period~$5 + \varepsilon$ in the instance $(\frac{3}{2}, 5 +
  \varepsilon, 9)$.

  Fix $\varepsilon > 0$ arbitrarily.  As in the proof of
  Lemma~\ref{lem:miyagi_2}, we denote by $P_{i}(y, l)$ the proposition
  that for each $m \in \mathbb{Z}$, there exist at least $l$ values of
  $t \in [m, m + y) \cap \mathbb{Z}$ such that
    $S(t) = i$, which means that the schedule $S$ performs task $i$ at
    least once every $y$ consecutive days.
    Clearly, for $y \leq y'$ and $l \geq l'$, $P_{i}(y, l)
    \Rightarrow P_{i}(y', l')$ holds true.  In addition, if $P_{i}(y,
    l)$ and $P_{i}(y', l')$ hold true, then so does $P_{i}(y + y', l +
    l')$.

  The schedulability of task~2 in the instance $(\frac{3}{2}, 5 +
  \varepsilon, 9)$ is equivalent to $P_{2}(\lceil l \cdot (5 +
  \varepsilon) \rceil, l)$ for all $l \in \mathbb{N}$. We are going to
  prove this.  We know from \eqref{eq:miyagi_3_2} that $S$ is itself
  has a period of 9 and $S(t) = 2$ for $t\equiv 2, 5$ $(\text{mod }
  9)$, which implies that $P_{2}(9, 2)$ is true.  See that $P_{2}(9,
  2)$ indicates that task~2 is performed approximately on average once
  every $\frac{9}{2} (< 5 + \varepsilon)$ days.  Roughly speaking,
  $P_{2}(\lceil l \cdot (5 + \varepsilon) \rceil, l)$ holds true for
  large $l$ and therefore we only need to check this for small $l$.
  We will often use the following inequality, $\lceil l \cdot (5 +
  \varepsilon) \rceil = 5l + \lceil l \varepsilon \rceil \geq 5l + 1$
  for any integer $l$.

  For $l = 1$, \eqref{eq:miyagi_3_2} says that there is at least one
  day $t$ satisfying $t \equiv 2, 5$ $(\text{mod } 9)$ in any $5l + 1
  = 6$ consecutive days, thus $P_{2}(6, 1) \Rightarrow P_{2}(\lceil 1
  \cdot (5 + \varepsilon) \rceil, 1)$.  For $l = 2$, since $\lceil 2
  \cdot (5 + \varepsilon) \rceil \geq 5 \cdot 2 + 1 = 11 > 1 \cdot 9$,
  $P_{2}(9, 2) \Rightarrow P_{2}(\lceil 2 \cdot (5 + \varepsilon)
  \rceil, 2)$ holds.  For $l = 3$, it follows that $\lceil 3 \cdot (5
  + \varepsilon) \rceil \geq 5 \cdot 3 + 1 = 16 = 1 \cdot 9 + 7$.
  Since $P_{2}(6, 1) \Rightarrow P_{2}(7, 1)$ and $P_{2}(9, 2)$, we
  have $P_{2}(9 + 7, 2 + 1) \Rightarrow P_{2}(\lceil 3 \cdot (5 +
  \varepsilon) \rceil, 3)$.  For $l = 4$, $\lceil 4 \cdot (5 +
  \varepsilon) \rceil \geq 5 \cdot 4 + 1 = 21 > 2 \cdot 9$ holds.
  Using $P_{2}(9, 2)$, we derive $P_{2}(2 \cdot 9, 2 \cdot 2)
  \Rightarrow P_{2}(\lceil 4 \cdot (5 + \varepsilon) \rceil, 4)$.  For
  $l = 5$, we have $\lceil 5 \cdot (5 + \varepsilon) \rceil \geq 5
  \cdot 5 + 1 = 26 = 2 \cdot 9 + 8$.  From $P_{2}(6, 1) \Rightarrow
  P_{2}(8, 1)$ and $P_{2}(9, 2)$, we get $P_{2}(2 \cdot 9 + 8, 2 \cdot
  2 + 1) \Rightarrow P_{2}(\lceil 5 \cdot (5 + \varepsilon) \rceil,
  5)$.  For $l = 6$, $\lceil 6 \cdot (5 + \varepsilon) \rceil \geq 5
  \cdot 6 + 1 = 31 > 3 \cdot 9$ holds.  We have $P_{2}(3 \cdot 9, 3
  \cdot 2) \Rightarrow P_{2}(\lceil 6 \cdot (5 + \varepsilon) \rceil,
  6)$ by $P_{2}(9, 2)$.  Finally, for $l \geq 7$, we have $\lceil l
  \cdot (5 + \varepsilon) \rceil \geq 5 l + 1 \geq 9 \cdot \lfloor
  \frac{5l + 1}{9} \rfloor$.  It also holds true that $2 \cdot \lfloor
  \frac{5l + 1}{9} \rfloor \geq l$, since $2 \cdot \lfloor \frac{5l +
    1}{9} \rfloor - l > 2 \cdot (\frac{5l + 1}{9} - 1) - l = \frac{l -
    7}{9} \geq 0$.  Applying $P_{2}(9, 2)$, we conclude $P_{2}(9 \cdot
  \lfloor \frac{5l + 1}{9} \rfloor, 2 \cdot \lfloor \frac{5l + 1}{9}
  \rfloor) \Rightarrow P_{2}(\lceil l \cdot (5 + \varepsilon) \rceil,
  l)$.
\end{proof}

Lemma~\ref{lem:real_3_cover} below is a helper lemma that derives a
set of schedulable real-valued instances from those found to be
schedulable by Lemma~\ref{lem:miyagi_3_2-7}, which is the counterpart
of Lemma~\ref{lem:real_2_cover} in Section~\ref{sec:2_periods_proof}.
Note that the resulting set includes instances whose periods are not
in non-decreasing order.  Lemma~\ref{lem:real_3_cover} follows
immediately from Lemma~\ref{lem:monotonicity}.
\begin{lemma}
  If a schedule $S$ is valid for the real-valued instance $(b_{1}, b_{2},
  b_{3})$, then $S$ is valid also for the real-valued instance $(a_{1},
  a_{2}, \frac{1}{\frac{5}{6} - \frac{1}{a_{1}} - \frac{1}{a_{2}}})$
  with $a_{1} \geq b_{1}$, $a_{2} \geq b_{2}$, and
  $\frac{1}{\frac{5}{6} - \frac{1}{a_{1}} - \frac{1}{a_{2}}} \geq
  b_{3}$.
  \label{lem:real_3_cover}
\end{lemma}
The conditions of cases~(II) through (VII) of
Lemma~\ref{lem:real_3_exact} turn out to be the conditions that are
fulfilled by schedulable real-valued instances obtained by applying
Lemma~\ref{lem:real_3_cover} to the instances found to be schedulable
by Lemma~\ref{lem:miyagi_3_2-7}.

Lemma~\ref{lem:real_3_exact} holds if the union of such real-valued
instances and real-valued instances that apply to case~(I), which are
shown to be schedulable by Lemma~\ref{lem:miyagi_1},
includes all real-valued instances that have three periods and a density of exactly
$\frac{5}{6}$.
To facilitate this discussion,
we map a projection of a real-valued instance $(a_{1}, a_{2}) \in
\{(x, y) \;|\; 1 \leq x, 1 \leq y\}$ by a bijection $(a_{1}, a_{2})
\mapsto (\frac{1}{a_{1}}, \frac{1}{a_{2}}) \in \{(X, Y) \;|\;
0 < X \leq 1, 0 < Y \leq 1\}$.
In other words, we look at \emph{frequencies} rather
than periods.

\begin{figure}
  \centering
  \begin{minipage}[b]{0.48\textwidth}\centering%
    \includegraphics{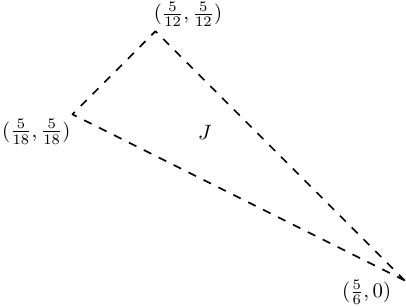}
  \end{minipage}\hfill
  \begin{minipage}[b]{0.48\textwidth}\centering%
    \includegraphics{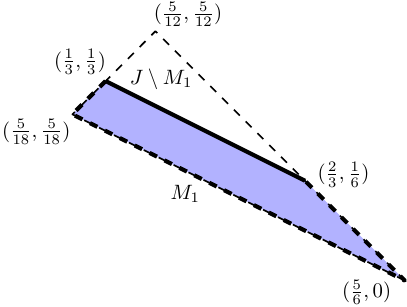}
  \end{minipage}
  \begin{minipage}[t]{0.48\textwidth}\centering%
    \caption{
      Illustration of the region $J$. The dashed lines are excluded.
      Each point in $J$ corresponds to a real-valued instance
      with three periods and a density of exactly $\frac{5}{6}$.}
    \label{fig:uragaeshi_j}
  \end{minipage}\hfill
  \begin{minipage}[t]{0.48\textwidth}\centering%
    \caption{
      Illustration of the region $M_{1}$. The dashed lines are excluded.}
    \label{fig:uragaeshi_1}
  \end{minipage}
\end{figure}
\begin{figure}
  \centering
  \begin{minipage}[b]{0.48\textwidth}\centering%
    \includegraphics{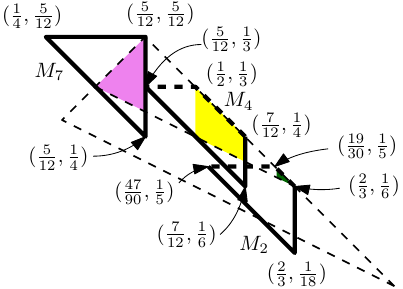}
  \end{minipage}\hfill
  \begin{minipage}[b]{0.48\textwidth}\centering%
    \includegraphics{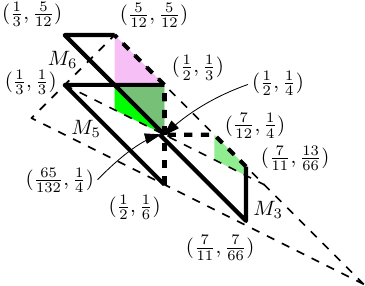}
  \end{minipage}
  \begin{minipage}[t]{0.48\textwidth}\centering%
    \caption{
      Illustration of the regions
      $M_{2}$, $M_{4}$, and $M_{7}$. The dashed lines are excluded.
      The colored regions indicate subregions of $J \setminus M_{1}$ which
      are each covered by $M_{2}$, $M_{4}$, and $M_{7}$.}
    \label{fig:uragaeshi_2_4_7}
  \end{minipage}\hfill
  \begin{minipage}[t]{0.48\textwidth}\centering%
    \caption{
      Illustration of the regions
      $M_{3}$, $M_{5}$, and $M_{6}$. The dashed lines are excluded.
      The colored regions indicate subregions of $J \setminus M_{1}$ which
      are each covered by $M_{3}$, $M_{5}$, and $M_{6}$.}
    \label{fig:uragaeshi_3_5_6}
  \end{minipage}
\end{figure}
The image of the whole set of real-valued instances
with three periods and a density of exactly $\frac{5}{6}$ is
\begin{equation}
  J := \Bigl\{(X, Y) \;\Big|\;
  X > Y, \frac{5}{6} < X + 2 Y, X + Y < \frac{5}{6}\Bigr\}.
\end{equation}
(Refer to the beginning of the statement of
Lemma~\ref{lem:real_3_exact} for the condition.)  The images of the
whole set of real-valued instances that apply to cases~(I) through
(VII) are, respectively,
\begin{alignat*}{2}
  M_{1} & := \Bigl\{(X, Y) \;\Big|\;
  X > Y, \frac{5}{6} < X + 2 Y \leq 1, X + Y < \frac{5}{6}\Bigr\},\\
  M_{2} & := \Bigl\{(X, Y) \;\Big|\;
  X \leq \frac{2}{3}, Y < \frac{1}{5},
  \frac{5}{6} - \frac{1}{9} \leq X + Y < \frac{5}{6}\Bigr\},\\
  M_{3} & := \Bigl\{(X, Y) \;\Big|\;
  X \leq \frac{7}{11}, Y < \frac{1}{4},
  \frac{5}{6} - \frac{1}{11} \leq X + Y < \frac{5}{6}\Bigr\},\\
  M_{4} & := \Bigl\{(X, Y) \;\Big|\;
  X \leq \frac{7}{12}, Y < \frac{1}{3},
  \frac{5}{6} - \frac{1}{12} \leq X + Y < \frac{5}{6}\Bigr\},\\
  M_{5} & := \Bigl\{(X, Y) \;\Big|\;
  X < \frac{1}{2}, Y \leq \frac{1}{3},
  \frac{5}{6} - \frac{1}{6} \leq X + Y < \frac{5}{6}\Bigr\},\\
  M_{6} & := \Bigl\{(X, Y) \;\Big|\;
  X < \frac{1}{2}, Y \leq \frac{5}{12},
  \frac{5}{6} - \frac{1}{12} \leq X + Y < \frac{5}{6}\Bigr\},
  \text{ and}\\
  M_{7} & := \Bigl\{(X, Y) \;\Big|\;
  X \leq \frac{5}{12}, Y \leq \frac{5}{12},
  \frac{5}{6} - \frac{1}{6} \leq X + Y < \frac{5}{6}\Bigr\}.
\end{alignat*}
See Figures~\ref{fig:uragaeshi_j}, 
\ref{fig:uragaeshi_1}, \ref{fig:uragaeshi_2_4_7}
and~\ref{fig:uragaeshi_3_5_6}.

Lemma~\ref{lem:real_3_inclusion} may be visually clear from the figures.
\begin{lemma}
  It holds that $J \subset \bigcup_{i = 1}^{7} M_{i}$.
  \label{lem:real_3_inclusion}
\end{lemma}
\begin{proof}
  First, we obtain $J \setminus M_{1} = \{(X, Y) \;|\; X > Y, 1 < X +
  2 Y, X + Y < \nicefrac{5}{6}\}$, which is the interior of the
  triangle $(\nicefrac{2}{3}, \nicefrac{1}{6})$, $(\nicefrac{5}{12},
  \nicefrac{5}{12})$, $(\nicefrac{1}{3}, \nicefrac{1}{3})$ (see
  Figures~\ref{fig:uragaeshi_j} and \ref{fig:uragaeshi_1}).  We will
  show that $J \setminus M_{1} \subset \bigcup_{i = 2}^{7} M_{i}$ by
  dividing into regions $\nicefrac{7}{11} < X \leq 1$,
  $\nicefrac{7}{12} < X \leq \nicefrac{7}{11}$, $\nicefrac{1}{2} \leq
  X \leq \nicefrac{7}{12}$, $\nicefrac{5}{12} < X < \nicefrac{1}{2}$,
  and $0 < X \leq \nicefrac{5}{12}$ (see
  Figures~\ref{fig:uragaeshi_2_4_7} and~\ref{fig:uragaeshi_3_5_6}).
  ~\\

  (i)~$(J \setminus M_{1}) \cap \{(X, Y) \;|\; \nicefrac{7}{11} < X \leq 1, 0
  < Y \leq 1\}$ is the interior of the triangle $(\nicefrac{2}{3},
  \nicefrac{1}{6})$, $(\nicefrac{7}{11}, \nicefrac{13}{66})$, $(\nicefrac{7}{11},
  \nicefrac{2}{11})$.  On the other hand, $M_{2} \cap \{(X, Y) \;|\;
  \nicefrac{7}{11} < X \leq 1, 0 < Y \leq 1\}$ is the sides and interior of the
  parallelogram $(\nicefrac{2}{3}, \nicefrac{1}{18})$, $(\nicefrac{2}{3},
  \nicefrac{1}{6})$, $(\nicefrac{7}{11}, \nicefrac{13}{66})$, $(\nicefrac{7}{11},
  \nicefrac{17}{198})$, excluding the side $(\nicefrac{7}{11},
  \nicefrac{13}{66})$, $(\nicefrac{7}{11}, \nicefrac{17}{198})$.  Noting that
  $\nicefrac{17}{198} < \nicefrac{2}{11} < \nicefrac{13}{66}$ for the $Y$
  coordinates of the vertices, we obtain $(J \setminus M_{1}) \cap
  \{(X, Y) \;|\; \nicefrac{7}{11} < X \leq 1, 0 < Y \leq 1\} \subset M_{2} \cap
  \{(X, Y) \;|\; \nicefrac{7}{11} < X \leq 1, 0 < Y \leq 1\}$
  (see Figure~\ref{fig:uragaeshi_2_4_7}).
  ~\\

  (ii)~$(J \setminus M_{1}) \cap \{(X, Y) \;|\; \nicefrac{7}{12} < X \leq
  \nicefrac{7}{11}, 0 < Y \leq 1\}$ is the sides and interior of the
  trapezoid $(\nicefrac{7}{11}, \nicefrac{2}{11})$, $(\nicefrac{7}{11},
  \nicefrac{13}{66})$, $(\nicefrac{7}{12}, \nicefrac{1}{4})$, $(\nicefrac{7}{12},
  \nicefrac{5}{24})$, excluding the three sides $(\nicefrac{7}{11},
  \nicefrac{13}{66})$, $(\nicefrac{7}{12}, \nicefrac{1}{4})$, $(\nicefrac{7}{12},
  \nicefrac{1}{4})$, $(\nicefrac{7}{12}, \nicefrac{5}{24})$, and $(\nicefrac{7}{12},
  \nicefrac{5}{24})$, $(\nicefrac{7}{11}, \nicefrac{2}{11})$.  $M_{3} \cap \{(X,
  Y) \;|\; \nicefrac{7}{12} < X \leq \nicefrac{7}{11}, 0 < Y \leq 1\}$ is the
  sides and interior of the parallelogram $(\nicefrac{7}{11},
  \nicefrac{7}{66})$, $(\nicefrac{7}{11}, \nicefrac{13}{66})$, $(\nicefrac{7}{12},
  \nicefrac{1}{4})$, $(\nicefrac{7}{12}, \nicefrac{7}{44})$, excluding the two
  sides $(\nicefrac{7}{11}, \nicefrac{13}{66})$, $(\nicefrac{7}{12}, \nicefrac{1}{4})$
  and $(\nicefrac{7}{12}, \nicefrac{1}{4})$, $(\nicefrac{7}{12}, \nicefrac{7}{44})$.
  Since $\nicefrac{7}{66} < \nicefrac{2}{11} < \nicefrac{13}{66}$ and
  $\nicefrac{7}{44} < \nicefrac{5}{24} < \nicefrac{1}{4}$, we have $(J \setminus
  M_{1}) \cap \{(X, Y) \;|\; \nicefrac{7}{12} < X \leq \nicefrac{7}{11}, 0 < Y
  < 1\} \subset M_{3} \cap \{(X, Y) \;|\; \nicefrac{7}{12} < X \leq
  \nicefrac{7}{11}, 0 < Y \leq 1\}$
  (see Figure~\ref{fig:uragaeshi_3_5_6}).
  ~\\

  (iii)~$(J \setminus M_{1}) \cap \{(X, Y) \;|\; \nicefrac{1}{2} \leq X
  \leq \nicefrac{7}{12}, 0 < Y \leq 1\}$ is the sides and interior of the
  trapezoid $(\nicefrac{7}{12}, \nicefrac{5}{24})$, $(\nicefrac{7}{12},
  \nicefrac{1}{4})$, $(\nicefrac{1}{2}, \nicefrac{1}{3})$, $(\nicefrac{1}{2},
  \nicefrac{1}{4})$, excluding the two sides $(\nicefrac{7}{12},
  \nicefrac{1}{4})$, $(\nicefrac{1}{2}, \nicefrac{1}{3})$ and $(\nicefrac{1}{2},
  \nicefrac{1}{4})$, $(\nicefrac{7}{12}, \nicefrac{5}{24})$.  $M_{4} \cap \{(X, Y)
  \;|\; \nicefrac{1}{2} \leq X \leq \nicefrac{7}{12}, 0 < Y \leq 1\}$ is the
  sides and interior of the parallelogram $(\nicefrac{7}{12},
  \nicefrac{1}{6})$, $(\nicefrac{7}{12}, \nicefrac{1}{4})$, $(\nicefrac{1}{2},
  \nicefrac{1}{3})$, $(\nicefrac{1}{2}, \nicefrac{1}{4})$, excluding the side
  $(\nicefrac{7}{12}, \nicefrac{1}{4})$, $(\nicefrac{1}{2}, \nicefrac{1}{3})$.  Due to
  $\nicefrac{1}{6} < \nicefrac{5}{24} < \nicefrac{1}{4}$, it holds that $(J
  \setminus M_{1}) \cap \{(X, Y) \;|\; \nicefrac{1}{2} \leq X \leq
  \nicefrac{7}{12}, 0 < Y \leq 1\} \subset M_{4} \cap \{(X, Y) \;|\;
  \nicefrac{1}{2} \leq X \leq \nicefrac{7}{12}, 0 < Y \leq 1\}$
  (see Figure~\ref{fig:uragaeshi_2_4_7}).
  ~\\

  (iv)~$(J \setminus M_{1}) \cap \{(X, Y) \;|\; \nicefrac{5}{12} < X <
  \nicefrac{1}{2}, 0 < Y \leq 1\}$ is the interior of the trapezoid
  $(\nicefrac{1}{2}, \nicefrac{1}{4})$, $(\nicefrac{1}{2}, \nicefrac{1}{3})$,
  $(\nicefrac{5}{12}, \nicefrac{5}{12})$, $(\nicefrac{5}{12}, \nicefrac{7}{24})$.
  $M_{5} \cap \{(X, Y) \;|\; \nicefrac{5}{12} < X < \nicefrac{1}{2}, 0 < Y <
  1\}$ is is the sides and interior of the trapezoid $(\nicefrac{1}{2},
  \nicefrac{1}{6})$, $(\nicefrac{1}{2}, \nicefrac{1}{3})$, $(\nicefrac{5}{12},
  \nicefrac{1}{3})$, $(\nicefrac{5}{12}, \nicefrac{1}{4})$, excluding the two
  sides $(\nicefrac{1}{2}, \nicefrac{1}{6})$, $(\nicefrac{1}{2}, \nicefrac{1}{3})$,
  and $(\nicefrac{5}{12}, \nicefrac{1}{3})$, $(\nicefrac{5}{12}, \nicefrac{1}{4})$.
  Besides, $M_{6} \cap \{(X, Y) \;|\; \nicefrac{5}{12} < X < \nicefrac{1}{2},
  0 < Y \leq 1\}$ is the sides and interior of the parallelogram
  $(\nicefrac{1}{2}, \nicefrac{1}{4})$, $(\nicefrac{1}{2}, \nicefrac{1}{3})$,
  $(\nicefrac{5}{12}, \nicefrac{5}{12})$, $(\nicefrac{5}{12}, \nicefrac{1}{3})$,
  excluding the two sides $(\nicefrac{1}{2}, \nicefrac{1}{4})$, $(\nicefrac{1}{2},
  \nicefrac{1}{3})$, and $(\nicefrac{5}{12}, \nicefrac{5}{12})$, $(\nicefrac{5}{12},
  \nicefrac{1}{3})$.  Note that $\nicefrac{1}{6} < \nicefrac{1}{4}$ and
  $\nicefrac{1}{4} < \nicefrac{7}{24}$ for the $Y$ coordinates of the
  vertices.  The part of region $(J \setminus M_{1}) \cap \{(X, Y)
  \;|\; \nicefrac{5}{12} < X < \nicefrac{1}{2}, 0 < Y \leq 1\}$ where $Y \leq
  \nicefrac{1}{3}$ is a subset of $M_{5} \cap \{(X, Y) \;|\; \nicefrac{5}{12}
  < X < \nicefrac{1}{2}, 0 < Y \leq 1\}$, while the rest of the region is a
  subset of $M_{6} \cap \{(X, Y) \;|\; \nicefrac{5}{12} < X < \nicefrac{1}{2},
  0 < Y \leq 1\}$.  Thus, $(J \setminus M_{1}) \cap \{(X, Y) \;|\;
  \nicefrac{5}{12} < X < \nicefrac{1}{2}, 0 < Y \leq 1\} \subset (M_{5} \cup
  M_{6}) \cap \{(X, Y) \;|\; \nicefrac{5}{12} < X < \nicefrac{1}{2}, 0 < Y <
  1\}$ (see Figure~\ref{fig:uragaeshi_3_5_6}).
  ~\\

  (v)~$(J \setminus M_{1}) \cap \{(X, Y) \;|\; 0 < X \leq
  \nicefrac{5}{12}, 0 < Y \leq 1\}$ is the sides and interior of the
  triangle $(\nicefrac{5}{12}, \nicefrac{7}{24})$, $(\nicefrac{5}{12},
  \nicefrac{5}{12})$, $(\nicefrac{1}{3}, \nicefrac{1}{3})$, excluding the two
  sides $(\nicefrac{5}{12}, \nicefrac{5}{12})$, $(\nicefrac{1}{3}, \nicefrac{1}{3})$
  and $(\nicefrac{1}{3}, \nicefrac{1}{3})$, $(\nicefrac{5}{12}, \nicefrac{7}{24})$.
  On the other hand, $M_{7} \cap \{(X, Y) \;|\; 0 < X \leq
  \nicefrac{5}{12}, 0 < Y \leq 1\}$ is the sides and interior of the
  triangle $(\nicefrac{5}{12}, \nicefrac{1}{4})$, $(\nicefrac{5}{12},
  \nicefrac{5}{12})$, $(\nicefrac{1}{4}, \nicefrac{5}{12})$.  Noting that
  $\nicefrac{1}{4} < \nicefrac{7}{24}$ for the $Y$ coordinates of the vertices
  and the point $(\nicefrac{1}{3}, \nicefrac{1}{3})$ lies on the side
  $(\nicefrac{1}{4}, \nicefrac{5}{12})$, $(\nicefrac{5}{12}, \nicefrac{1}{4})$, we
  obtain $(J \setminus M_{1}) \cap \{(X, Y) \;|\; 0 < X \leq
  \nicefrac{5}{12}, 0 < Y < 1\} \subset M_{7} \cap \{(X, Y) \;|\; 0 < X
  \leq \nicefrac{5}{12}, 0 < Y < 1\}$
  (see Figure~\ref{fig:uragaeshi_2_4_7}).
\end{proof}

\begin{proof}[of Lemma~\ref{lem:real_3_exact}]
  Lemmas~\ref{lem:miyagi_1}, \ref{lem:miyagi_3_2-7},
  and~\ref{lem:real_3_cover} imply that instances that apply to cases~(I)
  through (VII) are schedulable and
  a valid schedule is given either of
  by~\eqref{eq:miyagi_3_1}, \eqref{eq:miyagi_3_2}, \eqref{eq:miyagi_3_3},
  \eqref{eq:miyagi_3_4}, \eqref{eq:miyagi_3_5}, \eqref{eq:miyagi_3_6},
  or \eqref{eq:miyagi_3_7}.
  Lemma~\ref{lem:real_3_inclusion}
  states that any real-valued instance with three periods and density
  less than or equal to $\frac{5}{6}$ satisfies one of the conditions of
  cases~(I) through (VII).
\end{proof}

We continue the analysis of real-valued instances with three periods.
The proof of Theorem~\ref{thm:real_3_distinct} based on
Lemma~\ref{lem:real_3_exact} provides us with valid schedules for
real-valued instances with three periods and density less than or
equal to $\frac{5}{6}$.  For an instance that applies to case~(I) of
Lemma~\ref{lem:real_3_exact}, the schedule is determined by the
individual period values of the instance.  Namely, the schedule is
instance-sensitive, which did not occur in the case of two periods in
Section~\ref{sec:2_periods_proof}.  (The same is true if the
schedulability is derived from Lemmas~\ref{lem:monotonicity} and
\ref{lem:partitioning}, and Theorem~\ref{thm:real_2_density_1}.)

On the other hand, for an instance that applies to cases~(II) through
(VII), the schedule is independent of the individual period values of
the instance.  In particular, the last instance of
Lemma~\ref{lem:miyagi_3_2-7}, which corresponds to the condition of
case~(VII), has no $\varepsilon$, and the instance itself is the one
with minimum periods among those schedulable.  This fact can be
expressed in the context of \cite{doi:10.1137/1.9781611977042.8} as follows:
\begin{corollary}
  $\{((\frac{12}{5}, \frac{12}{5}, 6), |121213212123|)\}$ is a Pareto
  surface of the set of all real-valued instances $(a_{1}, a_{2},
  a_{3})$ such that $a_{1} \geq \frac{12}{5}$, $a_{2} \geq
  \frac{12}{5}$, $a_{3} \geq 6$, and $\frac{1}{a_{1}} +
  \frac{1}{a_{2}} + \frac{1}{a_{3}} \leq \frac{5}{6}$.
  \label{cor:pareto_3_partial}
\end{corollary}
Unlike Corollary~\ref{cor:pareto_2}, we conjecture that there is no
Pareto surface of the set of all real-valued instances with three periods
and density less than or equal to $\frac{5}{6}$.  In contrast, $\{((2,
4, 4), |1213|), (3, 3, 3), |123|)\}$ is known to be a Pareto surface
of the set of all \emph{integer-valued} instances with three
periods and density less than or equal to
$\frac{5}{6}$~(\cite{doi:10.1137/1.9781611977042.8}).

Rather than a Pareto surface, let us focus on a common valid schedule
for a given set of instances.  As we have seen in
Lemmas~\ref{lem:real_3_exact}, \ref{lem:miyagi_3_2-7}, and
\ref{lem:real_3_cover}, such a schedule is found for instances that
apply to each of cases~(II) through (VII).  Then, what about the
condition of case~(I)?  Note in Figures~\ref{fig:uragaeshi_2_4_7}
and~\ref{fig:uragaeshi_3_5_6} that almost half area of the trapezoid
$M_{1}$ is covered by $M_{2}, M_{3}, \ldots, M_{7}$, which means that
there is a finite set of valid schedules for instances that correspond
to the covered region.  However, we do not know anything about the
uncovered region.
\section{Future Work}
\label{sec:future_work}
Towards Conjecture~\ref{conj:real}, we would like to build a framework
that comprehensively handles real-valued instance sets, analogous to
the Pareto surface for integer-valued instance
sets, considered by \cite{doi:10.1137/1.9781611977042.8}.
Although our current
methodology could be extended to the case where the periods take four
or five distinct values, solving the general case would be difficult.

The method of \cite{10.1145/3618260.3649757,KawamuraPNAS} for deriving
Theorem~\ref{thm:real_2_density_1} seems hopeless to immediately
apply.  The method relies on the fact that: The subset of all
integer-valued instances such that the periods and the density are
less than a certain threshold is finite and therefore the
schedulability can be investigated using a computer.  However, a
subset of such real-valued instances is in general infinite.
\acknowledgements
\label{sec:ack}
We are grateful to Akitoshi Kawamura for useful discussions.
We also thank the editors and
anonymous reviewers for their valuable comments and suggestions.

\bibliographystyle{abbrvnat}
\bibliography{main}
\label{sec:biblio}

\end{document}